\documentclass[conference,10pt]{IEEEtran}
\usepackage{amsmath,amssymb,amsthm,graphicx,mathtools,enumerate,subcaption}
\usepackage[hidelinks]{hyperref}
\usepackage[capitalize,nameinlink,noabbrev]{cleveref}
\usepackage[margin=.5in]{geometry}

\newtheorem{lemma}{Lemma}
\newtheorem{theorem}{Theorem}

\theoremstyle{definition}

\newtheorem{defn}{Definition}[section]

\newcommand{\framed}[1]{\bigskip\noindent\fbox{\begin{minipage}{\linewidth-\fboxsep-\fboxsep}#1\end{minipage}}\bigskip}

\makeatletter
\newcommand{\clabel}[2]{%
  \protected@write \@auxout {}{\string \newlabel {#1}{{#2}{\thepage}{#2}{#1}{}}}%
  \raisebox{1em+\fboxsep}[0pt]{\hypertarget{#1}{}}}
\makeatother

\newcommand{\ineqref}[1]{\hyperref[#1]{\eqref{#1}}}

\title{Efficient Unbiased Sparsification}
\author{
\IEEEauthorblockN{Leighton Barnes\IEEEauthorrefmark{1}\IEEEauthorrefmark{2}, Stephen Cameron\IEEEauthorrefmark{1}, Timothy Chow\IEEEauthorrefmark{1}, Emma Cohen\IEEEauthorrefmark{1}, Keith Frankston\IEEEauthorrefmark{1}, \\
Benjamin Howard\IEEEauthorrefmark{1}, Fred Kochman\IEEEauthorrefmark{1}, Daniel Scheinerman\IEEEauthorrefmark{1}, and Jeffrey VanderKam\IEEEauthorrefmark{1}}
\vspace{.1cm}
\IEEEauthorblockA{\IEEEauthorrefmark{1} Center for Communications Research, Princeton, NJ 08540}
\IEEEauthorblockA{\IEEEauthorrefmark{2} corresponding author: l.barnes@idaccr.org}
}
\date{December 2023}

\newcommand{\R}{\mathbb{R}}
\newcommand{\E}{\mathbf{E}}

\newcommand{\cdist}{\mathcal{C}}

\DeclareMathOperator{\Div}{Div}
\DeclareMathOperator{\supp}{I}

\newcommand{\DKL}{D_{\mathrm{KL}}}

\newcommand{\mVert}{\mathrel{\Vert}}
\newcommand{\defeq}{\vcentcolon=}
\newcommand{\pinv}{permuta\-tion-invariant}
\newcommand{\ones}{\vec{1}}
\newcommand{\pd}[2][]{\tfrac{\partial #1}{\partial #2}}

\DeclarePairedDelimiter{\abs}{\lvert}{\rvert}
\DeclarePairedDelimiter{\norm}{\lVert}{\rVert}
\DeclarePairedDelimiter{\set}{\{}{\}}

\begin{document}

\maketitle

\begin{abstract}
 
An unbiased \emph{$m$-sparsification} of a vector $p\in \R^n$ is a random vector $Q\in \R^n$ with mean~$p$ that has at most $m<n$ nonzero coordinates. Unbiased sparsification compresses the original vector without introducing bias; it arises in various contexts, such as in federated learning and sampling sparse probability distributions. Ideally, unbiased sparsification should also minimize the expected value of a divergence function $\Div(Q,p)$ that measures how far away $Q$ is from the original $p$. If $Q$ is optimal in this sense, then we call it \emph{efficient}. Our main results describe efficient unbiased sparsifications for divergences that are either permutation-invariant or additively separable. Surprisingly, the characterization for permutation-invariant divergences is robust to the choice of divergence function, in the sense that our class of optimal $Q$ for squared Euclidean distance coincides with our class of optimal~$Q$ for Kullback--Leibler divergence, or indeed any of a wide variety of divergences.
\end{abstract}


\section{ Introduction}\label{sec:intro}

Suppose we have a vector $p\in\R^n$,
and (possibly because of memory or bandwidth limitations)
we want to approximate it with a vector $Q\in\R^n$
with at most $m$ nonzero entries, where $m < n$.
The construction of~$Q$ is allowed to be randomized,
and we want $Q$ to be the ``best possible approximation'' of~$p$.
``Best possible approximation'' will be defined precisely later,
but at minimum, we want the expected value of~$Q$ to equal~$p$.
Depending on our application, we may also desire the stronger property
that the sum of the entries of~$Q$ should always equal the sum of the entries of~$p$.
How should we construct~$Q$?

We call this task \emph{efficient unbiased
sparsification} (EUS)---sparsification, because the number of
nonzero entries is reduced from $n$ to~$m$;
unbiased, because the expected value of~$Q$ equals~$p$;
and efficient, in the statistical sense of diverging from~$p$
as little as possible.
The problem of efficient unbiased sparsification,
or something very close to it,
arises in several different contexts.

\begin{enumerate}
\item
In the context of \emph{sampling sparse probability distributions}, we have a discrete probability distribution $p$ on $n$ outcomes, and we seek to randomly construct a sparse probability distribution $Q$ that has at most $m$ nonzero probabilities. We would like this construction to be unbiased, in the sense that the average over the potential sparse probability distributions should be the original distribution $p$, and for it to minimize the expected statistical divergence between $Q$ and $p$. One natural choice of a divergence in this case would be the standard Kullback--Leibler (KL) divergence.
\item 
In \emph{distributed statistical estimation}~\cite{dist_est1,dist_est2,dist_est3,dist_est4}, statistical samples are distributed across a number of client nodes that must send bandwidth-limited messages to a central server. The central server then performs statistical analyses using the compressed versions of the samples. By sparsifying potentially high-dimensional samples, they can be communicated more efficiently to the central server. The constraint that the sparsification be unbiased preserves statistical properties such as the mean, while the efficiency of the sparsification ensures that the compressed versions of the samples are not too far from the original ones.

\item 
In the related field of \emph{federated learning}~\cite{fed_orig,fed_survey}, client nodes are trying to jointly train a machine learning model. In order to facilitate the distributed training, minibatch gradients or model updates need to be communicated between nodes so that they can be aggregated into a combined model. For large machine learning models, however, a single gradient vector can be billions of parameters long, and sparsification strategies can be deployed in order to reduce the associated communication cost \cite{rtopk}. Notably, the works 
\cite{konecny-richtarik, wang-et-al, wangni-wang-liu-zhang} describe optimization objectives and their associated optimal strategies that are similar to a special case of the present work. They consider a ``soft'' sparsification constraint, where the \emph{expected} number of nonzero components can be at most $m$, instead of the ``hard'' constraint that we use, and they consider only the squared Euclidean distance as their divergence.

Many works in this area use some form of sparsification in order reduce communication cost, and we refer the curious reader to the review paper \cite{fed_survey2} that gives an overview of some of these strategies and related issues such as quantization. Techniques such as ``top-$k$'' and ``rand-$k$'' that take the top components by magnitude or just randomly sample components have been shown to be effective. Furthermore, the works \cite{rtopk,horvath-richtarik} both show that a combination of these two strategies can outperform each one on its own separately.  The optimal algorithms that arise from our analysis in the present work have some elements of each of these strategies -- they keep components with sufficiently large magnitude and randomly subsample the others in a particular way. In this way, our work provides a principled reason that this combination is effective, and demonstrates concrete ways in which it can be optimal.

 In a slightly different optimization problem that occurs in federated learning, $n$ refers to the number of client nodes instead of the dimension of the updates, and 
bandwidth limitations force us to restrict the number~$m$
of clients allowed to communicate in each round.
Some clients have more important updates,
so the question arises of how to pick clients in a way that
respects their importance, while minimizing the
statistical distortion that restriction inevitably causes. This problem also leads to a similar optimization problem and is considered in \cite{chen-horvath-richtarik}.


\item 
In \emph{sampling with specified marginals}~\cite{tille},
the goal is to randomly choose a subset of $m$ items from
a population of $n>m$ items, in such a way that the
\emph{inclusion probability} of
item~$i$ ($1\le i\le n$)---i.e., the probability
that item~$i$ belongs to the chosen $m$-element subset---is
proportional to some specified positive number~$p_i$.
It turns out that for some sets of numbers~$p_i$,
it is impossible to achieve this goal exactly,
but we would still like to come as close as possible.
\end{enumerate}

In this paper, we solve the EUS problem by setting up the associated optimization problems, explicitly giving algorithms that produce optimal random vectors $Q$, and, in some cases, by describing the distributions of all random vectors $Q$ that optimize the objectives. We consider both permutation-invariant and additively separable divergences, which we will define shortly. Surprisingly, the characterization for permutation-invariant divergences is robust to the choice of divergence function, in the sense that
our class of optimal~$Q$ for (say) squared Euclidean distance
coincides with
our class of optimal~$Q$ for (say) Kullback--Leibler divergence,
or indeed any of a wide variety of divergences.

The space of unbiased $m$-sparsifications of a given~$p\in\R^n$
may be thought of as an infinite-dimensional
``simplex''~\cite[Section III.8]{barvinok},
which is a convex space in the sense that any mixture of
two unbiased $m$-sparsifications is an unbiased $m$-sparsification.
Expectation is linear, so we are minimizing a linear function
over a convex space.  That might sound promising,
but infinite-dimensional objects are not so easy to work with.
Indeed, it is not even obvious that any EUS exists. We proceed by reducing to a finite-dimensional, but not necessarily convex, problem. We then develop novel techniques for proving that our algorithms describe global minima for the finite-dimensional reductions.

\subsection{ Sampling with Specified Marginals}\label{sec:marginals}

In order to describe our algorithms, we will first need to define the notion of a \emph{heavy} index, and to this end
we will look more closely at the problem of 
sampling with specified marginals.
Assume first that
$\sum_i p_i = m$, and that $p_i \le 1$ for all~$i$.
Then there are many methods of randomly
sampling an \hbox{$m$-element} subset~$T$ of $\{1,2,\ldots,n\}$
in such a way that the inclusion probability of item~$i$ is
equal to (and not just proportional to)~$p_i$.
One method
is to partition a line segment of length~$m$
into $n$ subintervals such that the length of
subinterval~$i$ is~$p_i$, and then choose
$x\in[0,1)$ uniformly at random, and finally let
$i\in T$ if and only if $x+j$ lies in subinterval~$i$
for some integer~$j$.
We leave it to the reader to check that this method does indeed work.
Till\'e~\cite{tille} gives many other methods,
including one that maximizes entropy.

Now, let us re-examine the assumptions that
$\sum_i p_i = m$ and $p_i \le 1$ for all~$i$.
Given a probability distribution on $m$-element subsets,
let $s_i$ denote the inclusion probability of item~$i$.
Recall that in our original problem the marginal inclusion probabilities $s_i$ were only required to be
\emph{proportional} to~$p_i$, and not necessarily \emph{equal} to~$p_i$.
By linearity of expectation, the sum of the $s_i$
equals the expected total number of elements chosen---which
in this case is precisely~$m$.
It follows that for there to exist a probability distribution
on $m$-element subsets that achieves inclusion probabilities
proportional to the given numbers~$p_i$, a necessary condition
is that if the $p_i$ are rescaled so that they sum to~$m$,
then each rescaled $p_i$ must be \emph{equal} to~$s_i$,
and in particular must be at most~$1$.
In other words, if for any~$i$,
\begin{equation}
\label{eq:toobig}
p_i > \frac{1}{m} \sum_{j=1}^n p_j,
\end{equation}
then it is impossible to sample with the specified marginals.
What is one supposed to do in this case?

Of course, one option is to simply issue an error message
and give up.  However, Till\'e~\cite[Section 2.10]{tille}
offers a different approach.
If the largest $p_i$ is ``too big''---meaning that
it satisfies \ineqref{eq:toobig}---then
item~$i$ is automatically granted membership
in our chosen set of $m$ items.
We are thus reduced to choosing $m-1$ items from the
remaining set of $n-1$ candidates.  Again, if the largest
remaining $p_i$ is ``too big'' then it is automatically
included.  This process is iterated until we reach a set of~$p_i$
for which sampling with the specified marginals becomes possible.
Motivated by Till\'e's procedure, we make the following definition.

\begin{defn}\label{def:heavy}

Given a sequence of positive real numbers arranged (without loss of generality) in
weakly decreasing order $p_1 \ge p_2 \ge \cdots \ge p_n$,
and a positive integer $m<n$,
we say that index $i$ is \emph{$m$-heavy} if
\begin{equation}
\label{eq:mheavy}
\sum_{j=i+1}^n p_j \le (m-i)p_i.
\end{equation}
If $i$ is not $m$-heavy then we say it is \emph{$m$-light}.
\end{defn}

Adding $p_i$ to both sides of \ineqref{eq:mheavy}
shows that if $i$ is $m$-heavy, then
\begin{equation*}
\sum_{j=i}^n p_j \le (m-i)p_i + p_i = (m - (i-1))p_i \le (m-(i-1))p_{i-1};
\end{equation*}
i.e., $i-1$ is also $m$-heavy.
So there is some threshold $h$ up to which all the $i$ are $m$-heavy
and beyond which all the $i$ are $m$-light. In practice, the fastest way to locate this threshold is probably by binary search.

\subsection{ Efficient Unbiased Sparsification}

In order to state our main results, we must first give more precise definitions of unbiased sparsification and divergences.

\begin{defn}\label{def:sparsification}

Write $\supp(v) = \set{i \mid v_i \neq 0}$ for
the set of indices of nonzero coordinates of $v\in \R^n$, which we sometimes refer to as the \emph{survivor set}
of $v$.
If $p = (p_1, \ldots, p_n) \in \R^n$ and $m$ is a
positive integer such that $\abs{\supp(p)} > m$,
then a random vector $Q = (Q_1, \ldots, Q_n) \in \R^n$ is
an \emph{unbiased (random) \hbox{$m$-sparsification}} of~$p$ if
\begin{enumerate}[a)]
\item $\abs{\supp(Q)} \le m$ (i.e., $Q$ is $m$-sparse) and
\item $\E[Q] = p$ (i.e., $Q$ is an unbiased estimate of $p$).
\end{enumerate}
\end{defn}


When we say that an
unbiased 
\hbox{$m$-sparsification} $Q$ of~$p$ is \emph{efficient},
we mean that it minimizes $\E[\Div(Q,p)]$ among all unbiased 
$m$-sparsifications of $p$, where $\Div$
is a given \emph{divergence} function. 
But what exactly is a divergence function?
There are many different notions of divergence
in the literature~\cite{amari};
while we are not able to handle every such notion, our results cover two wide classes of functions.

\begin{defn}
 Let $X$ be a convex subset of~$\R^n$ and let
$\Div:X \times X \to \R$ be a function.
For fixed $p$, we write $D$ for the function $D(q) \defeq \Div(q,p)$.
\begin{enumerate}
\item $\Div$ is \emph{convex} if for every fixed $p$, $D$ is a convex function of $q$. We say $\Div$ is \emph{strictly convex} if
$D$ is twice differentiable\footnote{
``Twice differentiable'' means that the second-order Fr\'echet
derivative exists everywhere~\cite[Chapter VIII, Section~12]{dieudonne-analysis}.
In the literature, $F$ is not always assumed to be twice differentiable,
but then strange things can occur~\cite{dudley} that we prefer to ignore
in this paper. We should also emphasize that when we say that $\Div$
is ``strictly convex'', we require only that
$D$ is strictly convex for each fixed~$p$, and not that $\Div$ is
a strictly convex function jointly in $q$ and~$p$.}
and its Hessian matrix is positive definite everywhere.
\item $\Div$ is \emph{additively separable}
if for every fixed~$p$ there are
functions $f_1,\ldots,f_n$ (possibly depending on $p$) such that
\begin{equation}
\label{eq:additivelyseparable}
D(q) = \sum_{i=1}^n f_i(q_i).
\end{equation}
\item $\Div$ is \emph{\pinv}
if for every fixed $p$ there is a function $F\colon X \to \R$
and $\alpha \in \R^n$ (both possibly depending on $p$) such that
\begin{equation}
\label{eq:pinv}
D(q) = F(q) + \alpha \cdot q
\end{equation}
and $F$ is \pinv\ in~$q$; i.e., $F(q) = F(\sigma(q))$
where $\sigma(q) \defeq (q_{\sigma(1)}, \ldots, q_{\sigma(n)})$
is any vector obtained from~$q$ via a permutation~$\sigma$
of the coordinates.
\end{enumerate}
\end{defn}

Examples of divergences that are strictly convex,
additively separable, and \pinv\
include squared Euclidean distance
\begin{equation*}
\Div(q,p) = \norm{q-p}^2 = \sum_i (q_i - p_i)^2,
\end{equation*}
and the Kullback--Leibler divergence
\begin{equation*}
\Div(q,p) = \DKL(q \mVert p) = \sum_i q_i \log(q_i/p_i).
\end{equation*}
More generally, a \emph{Bregman divergence}~\cite{amari}
by definition has the form
\begin{equation}
\label{eq:bregman}
\Div(q,p) = G(q) - G(p) - (\nabla G(p))\cdot (q-p)
\end{equation}
for some strictly convex function~$G$ called the
\emph{Bregman generator}.  If $G$ is \pinv\ in $q$ (respectively, additively separable), then the Bregman divergence is \pinv\ (respectively, additively separable). This can be seen by setting $F(q) = G(q) - G(p) +  (\nabla G(p))\cdot p$ and $\alpha =  -\nabla G(p)$.

Similarly, any \emph{$f$-divergence}
\[\Div(q,p) = \sum_i q_i f(p_i/q_i)\]
(for convex $f$) is additively separable, but $f$-divergences are not typically \pinv. Note that our definition of \pinv\ may be somewhat counterintuitive; for instance, the (un-squared) Euclidean distance $\Div(q, p) = \norm{q-p}$ is \emph{not} \pinv\ by our definition (nor is it additively separable).

In order to ensure that $\E[\Div(Q,p)]$ makes sense for sparse $Q$, we require that $\Div$ be defined (and finite) on the (closed) orthant containing $p$. Our results still apply in many cases when the derivatives of the divergence are infinite on the boundary (as in the case of Kullback--Leibler divergence).

Our first main result, given in \autoref{sec:pinv}, is that for convex, \pinv\ divergences, efficient $m$-sparsifications of $p$ are given by the following simple algorithm. Note that in the following we assume $p_i > 0$. This is a natural assumption if $p$ is a probability distribution, but may not make as much sense when $p$ represents a gradient vector in the federated learning example. In this case, the method can still easily be applied by switching the sign of the negative $p_i$ components, and then after sparsification, switching the corresponding signs of $q_i$. This can be done without loss of generality provided that $\Div(q,p)$, with the signs of both $p_i$ and $q_i$ switched, is still a divergence that satisfies the required axioms. This is trivially the case with squared Euclidean distance.

\framed{\clabel{alg:usa}{US-PI}%
 \textbf{Algorithm.} \emph{Unbiased Sparsification for Permutation-Invariant Divergences (\ref*{alg:usa}).}\\
 Without loss of generality, reorder the coordinates
of $p\in\R_{>0}^n$
so that $p_1 \ge \cdots \ge p_n > 0$.
Let $H = \set{1, 2, \ldots, h}$ be the $m$-heavy indices and let
\begin{equation*}
l := \frac{1}{m-h} \sum_{j=h+1}^n p_j.
\end{equation*}
Sample $m-h$ of the indices $\{h+1,h+2,\ldots,n\}$
with specified marginals proportional to $p_i$ ($h+1\le i\le n$),
and call the sample set~$I$.
Return the vector $Q\in\R^n$ whose coordinates are given by
\begin{equation*}
Q_i = \begin{cases}
  p_i, & \text{if $i\le h$;} \\
  l,   & \text{if $i\in I$;} \\
  0,   & \text{otherwise.}
\end{cases}
\end{equation*}
}

\autoref{fig:usa} illustrates the possible $2$-sparsifications of $p\in \R^3_{>0}$ that may result from \autoref{alg:usa}. In (a), $p$ has no heavy indices and so \autoref{alg:usa} yields three possible values for the sparsifications of $p$: $(1/2,1/2,0)$, $(1/2,0,1/2)$, and $(0,1/2,1/2)$.  In (b), the first index of $p$ is heavy, and so the algorithm yields only two possible values for the sparsifications of $p$: $(p_1,1-p_1,0)$ and $(p_1,0,1-p_1)$.

This characterization has two remarkable features: first, the efficient sparsifications are \emph{independent} of the divergence that is being optimized, beyond its convexity and permutation-invariance; second, we find that the random variable $Q$ satisfies $\sum_i Q_i = \sum_i p_i$. Therefore if $p$ is a probability distribution, so is $Q$.

Our second main result, given in \autoref{sec:separable}, is a similar characterization of the efficient sparsifications of $p\in \R^n$ in the case where the divergence is strictly convex and additively separable (but not necessarily \pinv). In this case, the efficient sparsifications \emph{do} depend on the choice of divergence, since they are not constrained to having the same divergence function in each coordinate, and they do not typically satisfy $\sum_i Q_i = \sum_i p_i$. In making the strict convexity assumption instead of normal convexity, we are able to characterize all possible efficient unbiased sparsifications. This is also possible in \autoref{sec:pinv}, and the details can be found in Appendix \ref{sec:uniqueness}. Without strictness in \autoref{sec:separable}, it is also possible to prove efficiency without getting a complete characterization of all efficient sparsifications.

\section{ Permutation-Invariant Divergences}\label{sec:pinv}

Our first step is to use convexity to reduce the problem to \emph{concentrated} distributions, meaning they are supported on a finite (and bounded) number of values. This allows us to give a finite-dimensional parametrization of the problem.

Unfortunately, this parametrization is no longer convex as stated. We nevertheless show that critical points of the Lagrangian must correspond to efficient sparsifications.

Finally, we show that if $\Div$ is \pinv\ then the solutions corresponding to the random variable given by \ref{alg:usa} are indeed critical points, yielding our desired result.

Throughout this section we assume that $Q_i\geq 0$, i.e., that we only consider sparsifications that take nonnegative values (or more generally, $Q$ only takes values in the same orthant as $p$). This is done for technical reasons in order to facilitate the proof of Theorem \ref{thm:pinv}, but it may be possible to relax this assumption. Because of this, in Theorem \ref{thm:pinv} we only show that \autoref{alg:usa} produces sparsifications that are efficient among all \emph{nonnegative} sparsifications.

\subsection{ Facet Concentration}

Here we fix $p\in \R_{>0}^n$ and assume that $\Div(q, p) = D(q)$ is convex.
Throughout, the symbols $I$ and $J$ will represent subsets of $\{1,2,\ldots,n\}$ of cardinality $m$. For each $I$, let
\begin{align*}
  \Delta^I &= \set{x \in \R_{\geq 0}^n \mid \supp(x) = I} \\
  & = \set{x\in \R_{\geq 0}^n \mid x_i > 0 \text{ for } i\in I,\ x_i = 0 \text{ for } i\notin I}
\end{align*}
be the (open) facet consisting of the points whose nonzero coordinates are those with indices in $I$. We note that the $\Delta^I$ are pairwise-disjoint convex bodies. An unbiased sparsification\footnote{Note that here our $Q$ is implicitly constrained to have \emph{exactly} $m$ nonzero coordinates, whereas earlier we allowed \emph{at most} $m$ nonzero coordinates. This assumption serves to prevent the argument from becoming needlessly complicated. It can be shown that it is never optimal to use fewer than $m$ nonzero coordinates.} is a random variable $Q$ taking values in $\cup_I \Delta^I$ such that $\E[Q] = p$.

%


\begin{lemma}[Facet concentration]\label{lem:face-concentration}
  Assume that $\Div$ is convex, and let $Q$ be an $m$-sparsification of $p\in \R_{>0}^n$. For each $I$ with $\Pr(Q\in \Delta^I) > 0$, write $q^I \defeq \E[Q | Q \in \Delta^I]$. Then the sparsification $Q'$ such that $\Pr(Q'\in \Delta^I) = \Pr(Q\in\Delta^I)$ and $\Pr(Q' = q^I | Q'\in \Delta^I) = 1$ satisfies
  $\E[\Div(Q',p)] \leq \E[\Div(Q,p)].$ If $\Div$ is strictly convex and $Q,Q'$ do not have identical probability measures, then this inequality is strict.
  
\end{lemma}
We call such a $Q$ \emph{(facet-)concentrated}.
\begin{proof}
See Appendix \ref{pf_lem1}.
\end{proof}

We can parametrize the concentrated sparsifications in terms of the facet probabilities
\[x_I = \Pr(Q \in \Delta^I) \qquad x_I\in \R_{\geq 0}\]
along with the support points
\[y^I = \E[Q \mid Q\in \Delta^I] \qquad y^I\in \Delta^I.\]

So now we have a finite-dimensional problem:

\framed{\clabel{opt:dist}{SCDO}%
   \textbf{Problem.} \emph{Sparse Concentrated Distribution Optimization (\ref*{opt:dist}).}\\
   For convex $D: \R_{\geq 0}^n \to \R$,
  \begin{gather*}
  \text{minimize} \; \; f(x_I, y^I) \defeq \E[D(Q)] = \sum_I x_I D(y^I)  \; \; \text{subject to}\\
  \begin{alignedat}{2}
  x_I                        &\geq 0 \\
  S(x_I, y^I)           \defeq \sum_I x_I &= 1 \\
  G_i(x_I, y^I) \defeq \sum_{I} x_I y^I_i &= p_i &\quad& \text{for all $i$}\\
  y^I_i                         &= 0   &\quad& \text{for all } i \notin I\\
  \text{and} \qquad y^I_i        &\geq 0 &\quad& \text{for all $I, i$}.
  \end{alignedat}
  \end{gather*}
}

We write $Q \sim \cdist(x_I, y^I)$ to denote that the random variable $Q$ has the concentrated distribution corresponding to $(x_I, y^I)_{I}$, i.e., $Q$ takes on value $y^I\in \Delta^I$ with probability $x_I$. Note that there may be many choices of $(x_I, y^I)_I$ corresponding to the same random variable $Q$; in particular, if $x_I = 0$ then the choice of $y^I$ has no effect on the resulting distribution.

\subsection{ Optimality of Critical Points}

As written, the objective function of \ref{opt:dist} is convex, but its constraints are not. Therefore we cannot use techniques from convex optimization straight out of the box. Fortunately, in our case we find that a critical point of the Lagrangian still suffices to give a global optimum.

\begin{lemma}\label{thm:critical}
  Suppose that $D$ is smooth and that $Q$ is a concentrated unbiased sparsification of $p$.
  Suppose that for \emph{all} $(x_I,y^I)_I$ with $Q \sim \cdist(x_I, y^I)$
  (i.e., where $Q$ is $y^I\in \Delta^I$ with probability $x_I$)
  there exist $\nu, \lambda_i \in \R$,
  and $\mu_I \in \R_{\geq 0}$ such that 
  $$\nabla f(x_I, y^I) = \nu \nabla S(x_I, y^I) + \sum_{i=1}^n \lambda_i \nabla G_i(x_I, y^I) + \sum_{I: x_I = 0} \mu_I \nabla x_I.$$
  Then $Q$ is an efficient unbiased sparsification of $p$.
\end{lemma}

\begin{proof}
See Appendix \ref{pf_lem2}.
\end{proof}

\subsection{ Solving \ref*{opt:dist} for Permutation-Invariant Divergences}

Note that all of our results up to this point have only relied on the convexity of $D(q) = \Div(q,p)$. We now assume that $\Div$ is also \pinv, so that $D(q) = F(q) + \alpha\cdot q$, where $\alpha\in \R^n$ and $F(q) = F(\sigma(q))$ is invariant under permutations $\sigma$ of the coordinates of $q$. First note that for any unbiased sparsification $Q$ of $p$,
\[\E[D(Q)] = \E[F(Q)] + \alpha\cdot \E[Q] = \E[F(Q)] + \alpha\cdot p,\]
so it is equivalent to minimize $\E[F(Q)]$.

Recall that without loss of generality we're assuming $p_1 \geq p_2 \geq \cdots \geq p_n > 0$.
Below, we define a family of ``preservative'' unbiased sparsifications, and then later we will show that preservative unbiased sparsifications are efficient.

\newcommand{\q}{\tilde{y}}
\newcommand{\al}{\tilde{x}}
\begin{defn}
Let $H = \{1,2,\ldots,h\}$ be the set of $m$-heavy indices, as defined in \autoref{def:heavy}.
Let $\ell = \frac{1}{m-h}\sum_{i > h} p_i$, and recall that $p_h > \ell \geq p_{h+1}$ and $h < m$.
Denote
\[\q^I_i \defeq \begin{cases}
  p_i & i\in I\cap H\\
  \ell & i \in I\setminus H\\
  0 & i \notin I.
\end{cases}\]

We say a concentrated unbiased $m$-sparsification $Q$ of $p$ is \emph{preservative} if $Q \sim \cdist(\al_I, \q^I)$ for some $\al_I$ with $\al_I = 0$ for all $I\not\supset H$.
\end{defn}

Note that the $\al_I$ in a preservative unbiased $m$-sparsification of $p$ must satisfy the unbiasedness constraints
$\sum_{I\supset H} \al_I \q^I = p .$
Plugging in the definition of $\q^I$, we find that the unbiasedness constraints are equivalent to constraints on the marginal inclusion probabilities:
\begin{align*}
s_i \defeq \sum_{I \ni i} \al_I &= \begin{cases} 1 & i \in H\\ \frac{p_i}{\ell} & i\notin H. \end{cases}
\end{align*}
In other words, the preservative $Q$ are precisely those which are produced by \ref{alg:usa}: the random variable $I = \supp(Q)$ always contains all of the heavy indices and contains each light index $i$ with probability proportional to $p_i$.

\begin{theorem}\label{thm:pinv}
If $\Div$ is convex and \pinv, then the preservative unbiased $m$-sparsifications $Q$ of $p\in \R_{>0}^n$ are efficient (among all nonnegative sparsifications). If $\Div$ is strictly convex, then these are the only efficient $m$-sparsifications.  
\end{theorem}

\begin{proof}
See Appendices \ref{pf_thm1} through \ref{sec:uniqueness}.
\end{proof}

\section{ Additively Separable Divergences}\label{sec:separable}

Now we turn to the case where $D(Q) = \sum_i f_i(Q_i)$ is strictly convex and additively separable, but not necessarily \pinv. We also remove the constraint that the $Q_i$ have to be nonnegative (although we still write $p\in \R^n_{>0}$ without loss of generality.) Our proof takes a similar route to the \pinv\ case. We begin by strengthening \cref{lem:face-concentration} to
\emph{coordinate concentration}, meaning that if
$Q$ is an EUS (with no constraint on the sum of the coordinates of $Q$)
then for each coordinate~$i$ there is only one possible
nonzero value that $Q_i$ can take.
While \cref{lem:face-concentration} already allowed us to reduce our
problem to a finite (but exponentially large) number of dimensions, coordinate concentration reduces it further to a \emph{convex}\footnote{In fact, the function we find ourselves needing to optimize
has the form of an \emph{$f$-divergence}~\cite{amari}.
}\ optimization problem in $n$ variables---the inclusion probabilities $s_i$.

This allows us to solve this convex optimization problem
via a straightforward application
of Lagrange multipliers and the 
Karush--Kuhn--Tucker (KKT) conditions
\cite[Section 5.5.3]{boyd-vandenberghe}. The details of coordinate concentration and the remainder of the additively separable case are deferred to Appendix \ref{app:separable}.

We can characterize the efficient sparsifications of $p$ with respect to a separable divergence $\Div(q,p) = \sum_i f_i(q_i)$ as follows:

\framed{\clabel{alg:uwsa}{US-AS}%
   \textbf{Algorithm.} \emph{Unbiased Sparsification for Additively Separable Divergences (\ref*{alg:uwsa}).}\\
  Given $\Div(q,p) = D(q) = \sum_i f_i(q_i)$ strictly convex, define $g_i(x) = x f'_i(x) - f_i(x) + f_i(0)$. Let $\lambda > 0$ be the unique value such that
  \begin{equation*}
    \sum_i \min\left(1, \tfrac{p_i}{g_i^{-1}(\lambda)}\right) = m.
  \end{equation*}
  Declare index $i$ to be \emph{heavy} if $g_i(p_i) \geq \lambda$ (and \emph{light} otherwise), and let $h < m$ be the number of heavy indices.
  Sample $m-h$ of the light indices with specified marginals $s_i \defeq p_i/g_i^{-1}(\lambda) < 1$,
  and call the sample set $I$.
  Return the vector $Q\in\R^n$ whose coordinates are given by
  \begin{equation*}
    Q_i = \begin{cases}
      p_i, & \text{if $i$ is heavy;} \\
      g_i^{-1}(\lambda),   & \text{if $i\in I$;} \\
      0,   & \text{otherwise.}
    \end{cases}
  \end{equation*}
}

\begin{theorem}
\label{thm:separable}

Let $\Div$ be a strictly convex, additively separable divergence defined on~$\R^n$.
Then the efficient (with respect to $\Div$) unbiased $m$-sparsifications
of $p\in \R^n_{>0}$
are precisely those produced by \ref{alg:uwsa}.
\end{theorem}

In general the sparsifications produced by this procedure do \emph{not} satisfy the stronger constraint $\sum_i Q_i = \sum_i p_i$, but it is not hard to verify that if $D$ is also \pinv\ then this procedure aligns with \ref{alg:usa} and the resulting random variable does satisfy that constraint.

Note also that this result shows that (for additively separable divergences) an efficient sparsification of $p\in \R^n_{>0}$ must necessarily be nonnegative, whereas the proof of \autoref{thm:pinv} required this constraint to be enforced artificially.

\section{Generalizations and Open Questions}\label{sec:generalizations}

If we relax our definition of an $m$-sparsification from
requiring $\abs{\supp(Q)} \leq m$ to merely $\E[\abs{\supp(Q)}] \leq m$ (as in
\cite{chen-horvath-richtarik, konecny-richtarik, wang-et-al, wangni-wang-liu-zhang})
then we find that the efficient sparsifications are still characterized
by the same marginal inclusion probabilities $s_i$, albeit with more
leeway in the survivor sampling procedure.

What if our target vector $p$ is allowed to have negative (or complex, or other vector-valued) entries? For this question to make sense we must assume that $\Div$ is defined on the cone generated by the coordinates of $p$. For instance, for $p\in \R^{n\times k}$ it is perfectly sensible to ask (as in the second federated learning example in \autoref{sec:intro}) for the random variable $Q$ in $\R^{n\times k}$ with at most $m < n$ nonzero \emph{rows} which minimizes the expected squared Euclidean distance between $p$ and $Q$. We can always flip the signs on the $p$-inputs of $\Div(q, p)$ to treat $p$ as being in $\R_{\geq 0}^n$ and use the modified divergence in the optimization problem above. Flipping signs on the inputs preserves additive separability, so if $\Div$ is additively separable then \autoref{alg:uwsa} still applies. But flipping signs may destroy permutation-invariance, so even if $\Div$ is \pinv\ \autoref{alg:usa} may not apply. 

The federated learning literature is also interested in the case where the ``default'' value of $Q_i$ may be some nonzero value $z = (z_1,\dots,z_n)$ (typically the same for all $i$); that is, where $\supp(Q) \defeq \set{i: Q_i \neq z_i}$. In this case we can just replace $Q$, $p$, and $\Div$ with $\tilde{Q} = Q-z$, $\tilde{p} = p-z$, and $\widetilde{\Div}(\tilde{q}, \tilde{p}) = \Div(\tilde{q}+z, \tilde{p}+z) = \Div(q, p)$. Again, this transformation preserves additive separability but may destroy permutation-invariance and positivity of $p$. If one is going to allow a constant nonzero default value $z = (z_0,\dots, z_0)$ then it is also interesting to optimize the expected divergence over the choice of $z_0$ (for given, fixed $p$).

\begin{figure}
  \begin{centering}
    \begin{subfigure}{.2\textwidth}
      \includegraphics[width=4cm,page=1]{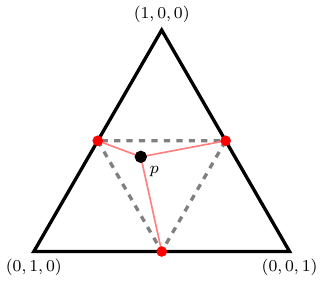}
      \caption{}
    \end{subfigure}%
    \hspace{1cm}
    \begin{subfigure}{.2\textwidth}
      \includegraphics[width=4cm,page=2]{n3m2.pdf}
      \caption{}
    \end{subfigure}
    \caption{Illustration of the probability simplex when using \autoref{alg:usa} on a probability distribution with $n=3$ and $m=2$.}\label{fig:usa}
  \end{centering}
\end{figure}

There remain several open questions yet to be answered. We noted above that in the non-\pinv\ case imposing the additional constraint $\sum_i Q_i = \sum_i p_i$ will change the optimal solution. What is the new optimum?

We also noted that the proof in \autoref{sec:pinv} only shows that the output of \ref{alg:usa} is efficient for \pinv\ divergences among \emph{nonnegative} sparsifications, whereas the proof in \autoref{sec:separable} shows that the same output is efficient among \emph{all} sparsifications as long as the divergence is also additively separable. We conjecture that the additional condition of additive separability is not actually necessary here.

Finally, we are also interested in knowing the answer for divergences which are neither \pinv\ nor additively separable, for instance, the squared \emph{Mahalanobis distance} that is the Bregman divergence with $G(x) = x^T A x$ for a positive definite matrix $A$.


\bibliography{aux_bib}
\bibliographystyle{IEEEtran}

\clearpage
\appendices
\section{Proofs}
\subsection{Proof of Lemma \ref{lem:face-concentration}} \label{pf_lem1}
Suppose $Q\in \Delta$ is an unbiased sparsification of $p$. Let $\pi$ denote the concentration map; i.e. $\pi(Q)$ is the random
variable such that $\pi(Q) = \E[Q \mid Q \in \Delta^I] = q^I \in \Delta^I$ with
probability $\Pr(Q \in \Delta^I)$.
Then clearly $\E(\pi(Q)) = \E(Q) = p$ so, $\pi(Q)$ is also an unbiased sparsification of $p$. By Jensen's inequality, $\E[D(\pi(Q))] \leq \E[D(Q)]$, with a strict inequality if $\Div$ is strictly convex and $Q$ was not already facet concentrated.

\subsection{Proof of Lemma \ref{thm:critical}} \label{pf_lem2}
  By way of contradiction, suppose $Q\sim \cdist(\alpha_I, u^I)$ is a concentrated unbiased sparsification of $p$ which satisfies the premise of the theorem but fails to be efficient.
  That means, using Lemma \ref{lem:face-concentration}, there is
  some other concentrated unbiased sparsification $Q'\sim \cdist(\beta_I, v^I)$ where
  $\E[D(Q')] < \E[D(Q)]$.
  
  Since we need only find one contradictory $(\alpha_I, u^I)$, we choose to take $u^I = v^I$ whenever $\alpha_I = 0$ or $\beta_I = 0$. That is, for any facet $\Delta^I$ where $Q$ has probability $\alpha_I = 0$ of appearing, we set the corresponding point $u^I$ (which value $Q$ never actually takes) equal to the point $v^I$ which $Q'$ may (or may never) take in $\Delta^I$, and vice versa. (If both $\alpha_I = \beta_I = 0$, pick $u^I = v^I\in \Delta^I$ arbitrarily.)

  For $0 \leq t \leq 1$, let
  $Q_t$ be the random variable corresponding to the convex mixture of distributions of $Q$ and $Q'$, where for any measurable set $A \subset \cup_I \Delta^I$ 
  we have
  $$\Pr(Q_t \in A) = (1-t) \Pr(Q \in A) + t \Pr(Q' \in A).$$
  Thus $Q_0 = Q$ and $Q_1 = Q'$. Note that since $Q$ and $Q'$ are both unbiased sparsifications of $p$, it is clear that $Q_t$ is also an unbiased sparsification of $p$.
  When $0 < t < 1$, the random variable $Q_t$ is not necessarily concentrated, as it can take on up to two different values in $\Delta^I$ rather than just one.  Let $g(t) \defeq \E[D(Q_t)]$, and note that
  $$g(t) = (1-t) \E[D(Q)] + t \E[D(Q')]$$
  is an affine function of $t$.
  In particular, $g'(0) = \E[D(Q')] - \E[D(Q)] < 0$.
  
  Let $\pi(Q_t)$ denote the concentration of $Q_t$ as above, and let $h(t) \defeq \E[D(\pi(Q_t))]$.
  We have $h(t) \leq g(t)$ for all $t$ by convexity of $D$.
  Also note that $h(0) = g(0)$ and $h(1) = g(1)$, since both $Q$ and $Q'$ are already concentrated.
  
  Now we show that $h : [0,1] \to \R$ is a smooth function.
  Since $Q_t$ takes value $u^I$ with probability $(1-t)\alpha_I$ and $v^I$ with probability $t \beta_I$, we can calculate that $\pi(Q_t) \sim \cdist(x_I(t), y^I(t))$, where
  \begin{align*}
    x_I(t) &:= (1-t) \alpha_I + t \beta_I,\\
    y^I(t) &:= \begin{cases}
      \frac{(1-t) \alpha_I u^I + t \beta_I v^I}{(1-t) \alpha_I + t \beta_I} &\text{if } \alpha_I, \beta_I > 0\\
      v^I = u^I & \text{if }\alpha_I = 0 \text{ or } \beta_I = 0.
      \end{cases}
  \end{align*}

  It is clear that $x_I(t)$ and $y^I(t)$ are smooth in $t$ for any given $I$.
  We have
  \begin{equation*}
    h(t) = \E[D(\pi(Q_t))] = \sum_{I : \alpha_I + \beta_I > 0} x_I(t) D(y^I(t)).
  \end{equation*}
  Since $D$ is smooth and the maps $t \mapsto x_I(t)$ and $t \mapsto y^I(t)$ are smooth, we know that $h$ is smooth.

  Now,
  \begin{align*}
    h'(0) &= \lim_{t \to 0^+} {h(t) - h(0) \over t} \\
    &= \lim_{t \to 0^+} {h(t) - g(0) \over t} \quad \text{(since $h(0) = g(0)$)} \\
    &\leq \lim_{t \to 0^+} {g(t) - g(0) \over t} \quad \text{(since $h(t) \leq g(t)$ and $t > 0$)} \\
    &= g'(0) \\
    &= \E[D(Q')] - \E[D(Q)] \\
    &< 0.
    \end{align*}

  Since $\gamma(t) := (x_I(t), y^I(t))_I$ is smooth, 
  we can apply the chain rule:
  \begin{align*}
    h'(0) &= (f \circ \gamma)'(0) \\
    &= (\nabla f)(\gamma(0)) \cdot \gamma'(0) \\
    &= (\nabla f)((\alpha_I, u^I)_I) \cdot \gamma'(0) \\
    &= \left(\nu \nabla S + \sum_{i=1}^n \lambda_i \nabla G_i + \sum_{I : \alpha_I=0} \mu_I \nabla x_I\right) \cdot \gamma'(0).
  \end{align*}
  But $\gamma(t)$ satisfies all the constraints, and hence
  $\gamma'(0)$ is perpendicular to the gradients $\nabla S$ and $\nabla G_j$.
  Furthermore, we also have that $\gamma'(0)$ is nonnegative 
  on all the $x_I$ components for which $\alpha_I = 0$, due to the
  inequality constraints $x_I \geq 0$.  Since every such $\mu_I \geq 0$,
  we conclude that $(\mu_I \nabla x_I) \cdot \gamma'(0) \geq 0$. 
  Hence $h'(0) \geq 0$, but this contradicts   
  $h'(0) < \E[D(Q')] - \E[D(Q)] < 0$.

\subsection{Proof of Theorem \ref{thm:pinv}} \label{pf_thm1}
 For now, assume that $F$ is smooth. Following \cref{thm:critical}, for every $(x_I, y^I)_I$ corresponding to a preservative sparsification $Q$, we will find multipliers $\nu, \lambda_i\in \R$ and $\mu_I \in \R_{\geq 0}$ such that
  \[\nabla f(x_I, y^I) = \nu \nabla S(x_I, y^I) + \sum_{i=1}^n \lambda_i \nabla G_i(x_I, y^I) + \sum_{I:x_I = 0} \mu_I \nabla x_I,\]
  where
  \begin{align*}
    f(x_I,y^I) &= \sum_{I} x_I F(y^I),\\
    S(x_I,y^I) &= \sum_I x_I,\\
\text{and} \quad G_i(x_I,y^I) &= \sum_{I \ni i} x_I y^I_i
  \end{align*}
as in \ref{opt:dist}. Note that if $Q\sim \cdist(x^I, y^I)$ is preservative then $y^I = \q^I$ whenever $x_I > 0$. It is straightforward to compute the various gradients:
\[\begin{array}{r|cccc}
& f(x_I, y^I) & S(x_I, y^I) & G_j(x_I, y^I) & x_J \\\hline
\pd{x_I} & F(y^I) & 1 & y^I_j & \delta_{J = I}\\
\pd{y_i^I} & x_I \pd[F]{q_i}(y^I) & 0 & x_I \delta_{j=i} & 0 \\
\end{array}\]
Let $I_0 = \{1,2,\ldots,m\} \supset H$, so that \vspace{-1em}
\[\q^{I_0} = (\overbrace{p_1, p_2, \ldots, p_h \vphantom{\ell}}^{h}, \overbrace{\ell, \ell, \ldots, \ell}^{m-h}, \overbrace{0, 0, \ldots, 0}^{n-m}).\]
Note that $p_h \geq \ell$. Because $F$ is \pinv, we have $\pd[F]{q_i}(\q^{I_0}) = \pd[F]{q_j}(\q^{I_0})$ whenever $\q^{I_0}_i = \q^{I_0}_j$. Therefore we can write \vspace{-1em}
\[\nabla F(\q^{I_0}) = (\overbrace{a_1, a_2 \dots, a_h\vphantom{b}}^{h}, \overbrace{b, b, \dots, b}^{m-h}, \overbrace{c, c, \dots, c\vphantom{b}}^{n-m}).\]

Now we determine the $\lambda_j$.
A brief look at the gradient table shows that we must have
\[\lambda_j = \pd[F]{q_j}(\q^I) \quad\text{ if $j \in I$ and $x_I > 0$}.\]
In particular, as $I$ varies over those $m$-element sets containing $j$
where $x_I > 0$, we need $\pd[F]{q_j}(\q^I)$ to be constant.
Indeed, for any $I$ with $x_I > 0$,
\[\pd[F]{q_j}(\q^I) = \begin{cases} a_j & j \in H \subset I \\ b & j\in I\setminus H \\ 0 & j\notin I
\end{cases}\]
does not depend on $I\ni j$.
Hence we have
\[\lambda = (\overbrace{a_1,a_2,\ldots,a_h\vphantom{b}}^h,\overbrace{b,b,\ldots,b}^{n-h}).\]

With this $\lambda$ we have matched $\nabla f$ on the $\pd{y^I_i}$ components, but we have not yet matched $\nabla f$ on the $\pd{x_I}$ components.  For that, we will need to use a multiple $\nu$ of $\nabla S$ as well as nonnegative multiples $\mu_I$ of the $\nabla x_I$ wherever $x_I = 0$.

We want to show that if $(x_I, y^I)_I$ is preservative then there exist $\nu$ and $\mu_I \geq 0$ (with $\mu_I = 0$ whenever $x_I > 0$) such that:
\[F(y^I) = \nu + \lambda \cdot y^I + \mu_I.\]

If $x_I > 0$ then $\mu_I = 0$ and preservativity requires that $I\supset H$ and $y^I = \q^I$. Permutation-invariance tells us that $F(\q^I) = F(\q^{I_0})$ and it is easy to see that $\lambda \cdot \q^I = \lambda \cdot \q^{I_0}$, so the constraints for $x_I > 0$ are satisfied by setting
\begin{align*}
  \nu = F(\q^{I}) - \lambda \cdot \q^{I} = F(\q^{I_0}) - \lambda \cdot \q^{I_0}.
\end{align*}

If $x_I = 0$ then $y^I \in \Delta^I$ is unconstrained by preservativity, and in order to have $\mu_I \geq 0$ we must show that
\begin{equation}
  F(y^I) - \lambda \cdot y^I \geq \nu = F(\q^{I_0}) - \lambda \cdot \q^{I_0}. \label{eqn:y-vs-q}
\end{equation}

If $I \supseteq H$, then $\lambda \cdot y^I = \nabla F (\q^I) \cdot y^I$ 
(because $\lambda$ and $\nabla F(\q^I)$ agree wherever $y^I$ is nonzero), so \ineqref{eqn:y-vs-q} holds by convexity of $F$:
\begin{gather*}
  F(y^I) \geq F(\q^I) + \nabla F(\q^I) \cdot (y^I - \q^I) \\
  = F(\q^I) + \lambda \cdot (y^I - \q^I)\\
\implies F(y^I) - \lambda \cdot y^I \geq F(\q^I) - \lambda \cdot \q^I = \nu.
\end{gather*}

It remains to show \ineqref{eqn:y-vs-q} for $I\not\supset H$. First note that $F$ is Schur convex since it is symmetric and convex.
Hence, 
$\pd[F]{q_i}(y) \geq \pd[F]{q_j}(y)$
whenever $y_i \geq y_j$.
Since the components of $\q^{I_0}$ are decreasing, we know that
the components of $\nabla F (\q^{I_0})$ are decreasing:
$$a_1 \geq \cdots \geq a_h \geq b \geq c.$$
In particular, $\lambda$ is decreasing.

Let $\sigma$ be a permutation such that $\sigma(y^I)$ is decreasing.
Then $\sigma(y^I) \in \Delta^{I_0}$.
Furthermore, since $\lambda$ is decreasing and $y^I \geq 0$, $\sigma$ is the permutation which maximizes
the inner product $\lambda \cdot \sigma(y^I)$.
Thus we have:
\begin{alignat*}{2}
  F(y^I) - \lambda \cdot y^I &= F(\sigma(y^I)) - \lambda \cdot y^I &\quad& \text{($F$ is perm.-inv.)} \\
  &\geq F(\sigma(y^I)) - \lambda \cdot \sigma(y^I) &\quad& \text{($\lambda$ is decreasing)} \\
  &\geq F(\q^{I_0}) - \lambda \cdot \q^{I_0} = \nu &\quad& \text{($\sigma(y)\in \Delta^{I_0}$).}
\end{alignat*}
We note for later use that if $F$ is strictly convex, then $\mu_I >0$ when $I\not\supset H$.  To see this, note that strict convexity implies $a_h>b$, and hence $\lambda\cdot y^I < \lambda \cdot \sigma(y^I)$.

We have met the premise of \cref{thm:critical}
and therefore shown that any preservative $Q$ is efficient.

\subsection{Removing the Smoothness Condition} \label{pf_nosmooth}
Now we turn the general case, where $F$ is convex and permutation-invariant but 
not necessarily smooth.
We note that any optimal $Q$ has the same total sum as $p$, hence 
we may restrict the domain of $F$ to the simplex 
$A = \{x \in \mathbb{R}_{\geq 0}^n \mid \sum_{i=1}^n x_i \leq \sum_{i=1}^n p_i\}$.

We first show that $F$ can be approximated by a smooth $\widetilde{F}$ 
on the domain $A$, where $\widetilde{F}$ is also convex and permutation-invariant. 
The main idea is to shrink the domain $A$ a little to $A'$, to give ``wiggle-room'', 
and then convolve $F$ with smooth density function $\theta$ with small support, 
yielding a smooth approximation $G$ to $F$ which is defined on $A'$.
The smooth approximation $G$ remains convex because it is a mixture of translates of $F$ which are all convex themselves. 
This portion of our argument is taken from \S 2 of \cite{ghomi}.
We then choose an affine contraction $R : A \to A'$ which only moves points a slight distance. 
We define $H(x) = G \circ R$, which is smooth and convex, and approximately equal to $F$.  However, $H$ is not permutation-invariant, so  
we define $\widetilde{F}(x) = {1 \over n!} \sum_{\sigma \in S_n} H(\sigma(x))$ to be the average of $H$ over all permutations.
The permutation-invariant $\widetilde{F}$ is convex, and is even closer to $F$ than was $H$, since $F$ is permutation-invariant.   

For $\delta > 0$ let $B_\delta(0) = \{y \in \mathbb{R}^n \mid \|y\|_2 < \delta\}$ be the ball of radius $\delta$ centered at $0$.
Let $A_\delta = \{x \in A \mid x + y \in A \text{ for all } y \in B_\delta(0)\}$.
Note that $A_\delta$ is convex.
Let $\delta_0$ be 
small enough so that $|F(x_1) - F(x_2)| < \epsilon/2$ whenever 
$x_1, x_2 \in A$ and $|x_1 - x_2| < \delta_0$.
Let $a \in A$ be the mean of $A$. 
Consider the affine contraction $R_t(x) = a + t(x-a)$ where $0 < t < 1$.
There exists some $t$ such that 
$\|R_{t}(x) - x\| < \delta_0$ for all $x \in A$.
In particular, we have $|F(R_t(x)) - F(x)| < \epsilon/2$ for all $x \in A$.
Furthermore, there exists $\delta_1 > 0$ such that 
the image of $R_t$ is contained in $A_{\delta_1}$.

Suppose that $A_\delta$ is nonempty (this is true if $\delta$ is sufficiently small).
Let $\theta$ be a smooth density function supported in $B_\delta(0)$.
For $x \in A_\delta$, define 
$G_\delta(x) = \int_{y \in B_\delta(0)} F(x-y) \theta(y) dy$.
Then $G_\delta$ is both convex and smooth.  Furthermore, 
$G_\delta(x) - F(x) = \left(\int_{y \in B_\delta(0)} F(x-y) \theta(y) dy \right) - F(x) = \int_{y \in B_\delta(0)} (F(x-y) - F(x)) \theta(y) dy$.
  In particular, we have that $|G_\delta(x) - F(x)| < \epsilon/2$ for all $x \in A_\delta$ 
for $\delta \leq \delta_0$.

Now take $\delta = \min(\delta_0, \delta_1)$. 
Suppose $x \in A$. We know that $R_t(x) \in A_{\delta_1} \subset A_\delta$.
Hence $R_t(x)$ lies in the domain of $G_\delta$.
We know that $|G_\delta(R_t(x)) - F(R_t(x))| < \epsilon/2$.
Furthermore, $|F(R_t(x)) - F(x)| < \epsilon/2$, and so 
$|G_\delta(R_t(x)) - F(x)| < \epsilon$.
Now $H := G_\delta \circ R_t$ is convex, smooth, and 
$|H(x) - F(x)| < \epsilon$ for all $x \in A$.
Finally we define $\widetilde{F}(x) = {1 \over n!} \sum_{\sigma \in S_n} H(\sigma(x)$.  Then $\widetilde{F}$ is convex, smooth, and permutation-invariant.  
Furthermore, $|\widetilde{F}(x) - F(x)| = {1 \over n!} \left|\sum_{\sigma \in S_n} H(\sigma(x)) - F(\sigma(x))\right| < \epsilon$. 

There is a sequence $\widetilde{F}_k : A \to \mathbb{R}$ such that 
each $\widetilde{F}_k$ is smooth, convex, and permutation-invariant, 
where $\widetilde{F}_k(x) \to F(x)$ as $k \to \infty$.   
 Suppose that $Q'$ is any unbiased sparsification. Let $Q$ be preservative; 
i.e. any $Q$ as defined by \ref{alg:usa}.  Then $Q$ is simultaneously optimal 
for all the $\widetilde{F}_k$.
The domain $A$ is compact and so $\E[\widetilde{F}_k(Q')] \to \E[F(Q')]$ 
and $\E[\widetilde{F}_k(Q)] \to \E[F(Q)]$ as $k \to \infty$.
Since $\E[\widetilde{F}_k(Q')] \geq \E[\widetilde{F}_k(Q)]$ for all $k$, 
we have $\E[F(Q')] \geq \E[F(Q)]$.  Hence $Q$ is optimal for $F$.

\subsection{Uniqueness under Strict Convexity} \label{sec:uniqueness}
Now assume that $F$ is strictly convex, and let $Q'$ be any efficient $m$-sparsification.  We will show that in fact $Q'$ is preservative.  Note that by Lemma \ref{lem:face-concentration}, we must have that $Q'$ is facet concentrated, so $Q'\sim \cdist(x_I', y^{I'})$. 

Let $Q \sim \cdist(x_I,y_I)$ be a preservative sparsification such that $x_I>0$ for all $I \supset H$.  For $0 \leq t \leq 1$, let
  $Q_t$ be the random variable corresponding to the convex mixture of distributions of $Q$ and $Q'$, where for any measurable set $A \subset \cup_I \Delta^I$ 
  we have
  $$\Pr(Q_t \in A) = (1-t) \Pr(Q \in A) + t \Pr(Q' \in A).$$
  Thus $Q_0 = Q$ and $Q_1 = Q'$. Note that $Q_t$ is still an unbiased sparsification, with 

  $$\E[D(Q_t)]= (1-t)\E[D(Q)] + t\E[D(Q')].$$
Hence $Q_t$ is efficient as well, and thus facet concentrated.  Then by the definition of $Q_t$, we must have that $y^{I'} = y^I$ whenever $I\supset H$ and $x_I'>0.$  To prove that $Q'$ is preservative, it thus remains to show that $x_I'=0$ for all $I\not\supset H$.  

Letting $x_I(t) = (1-t)x_I + tx_I'$, we have that $Q_t \sim\cdist(x_I(t),y^{I'})$ and that $f(x_I(t),y^{I'})$ is constant in time.  As $Q_0 = Q$, we have

\begin{equation}
    \begin{split}
    0 &= \frac{d}{dt}\bigg|_{t=0}f(x_I(t), y^{I'}) \\ 
    & =  \left(\nu \nabla S+ \sum_{i=1}^n \lambda_i \nabla G_i + \sum_{I: x_I = 0} \mu_I \nabla x_I\right) \cdot (x_I'-x_I,0) \\
    &=\sum_{I: x_I = 0} \mu_I  (x_I'-x_I) = \sum_{I: x_I = 0} \mu_I x_I'.
    \end{split}
\end{equation}
Thus as $\mu_I > 0$ for all $I\not \supset H$ by our earlier observation, it follows that $x_I' = 0$ for all such $I$.  Therefore $Q'$ is preservative.

\section{Additively Separable Divergences} \label{app:separable}
\subsection{ Coordinate Concentration}

\begin{lemma}[Coordinate concentration]
\label{lem:coord-concentration}

Assume that $\Div$ is strictly convex and additively separable.
Let $Q$ be an efficient unbiased $m$-sparsification of~$p \in \R^n_{>0}$.
Then for all~$i$, there exists $q_i > 0$
such that $\Pr(Q_i = q_i \mid Q_i \neq 0) = 1$.
That is,
any efficient unbiased $m$-sparsification of~$p$ is concentrated
on a unique nonzero value in each coordinate.
\end{lemma}

\begin{proof}
Consider any unbiased $m$-sparsification $Q$ of $p$.
For any~$i$,
if $\Pr(Q_i \ne 0) = 0$ then $p_i = \E[Q_i] = 0$,
whereas we are assuming that $p_i > 0$.
So it is legitimate to condition on $Q_i \ne 0$,
since this is an event with positive probability, and hence we can define
$q_i \defeq \E[Q_i \mid Q_i \ne 0]$. Note that $p_i = \E[Q_i] = q_i \Pr[Q_i \ne 0]$, so in particular $p_i > 0$ implies $q_i > 0$.

Fix any index~$i$, and let $Q'$ be the
random variable obtained as follows: first sample $\hat{q}\gets Q$; if
$\hat{q}_i = 0$ return $\hat{q}$; otherwise replace the $i$-th coordinate
of $\hat{q}$ with $q_i$ and return the result.
For simplicity of notation, in what follows, we assume without loss
of generality that $i=1$.

It is easy to check that $\E[Q'] = \E[Q]$
and $\abs{\supp(Q')} = \abs{\supp(\hat{q})} = m$, so $Q'$ is also an
unbiased $m$-sparsification of~$p$.
(But note that $\sum_j Q'_j$ is in general not equal to $\sum_j Q_j$, so even if $Q$ takes values in the probability simplex, $Q'$ usually does not).

By hypothesis, $\Div$ is additively separable.
Because $\Div$ is strictly convex,
the functions $f_i$ in \cref{eq:additivelyseparable}
are strictly convex.
By linearity of expectation,
\begin{align*}
& \E[D(Q') \mid Q_1 \neq 0] \\
  &= \E[D(q_1,Q_2, Q_3, \ldots ,Q_n) \mid Q_1 \neq 0] \\
  &= \E[f_1(q_1) + f_2(Q_2) + f_3(Q_3) + \cdots + f_n(Q_n)) \mid Q_1 \neq 0] \\
  &= \E[f_1(q_1) \mid Q_1 \ne 0] +\E[f_2(Q_2) \mid Q_1 \ne 0]  \\
  & \qquad {} + \cdots + \E[f_n(Q_n) \mid Q_1 \neq 0] \\
  &= f_1(\E[Q_1 \mid Q_1 \ne 0]) +\E[f_2(Q_2) \mid Q_1 \ne 0]  \\
  & \qquad {} + \cdots + \E[f_n(Q_n) \mid Q_1 \neq 0].
\end{align*}
Similarly,
\begin{align*}
\E[D(Q) \mid Q_1 \neq 0]
  = & \E[f_1(Q_1) \mid Q_1 \ne 0] \\
  & +\E[f_2(Q_2) \mid Q_1 \ne 0]  \\
  & + \cdots + \E[f_n(Q_n) \mid Q_1 \neq 0].
\end{align*}
Since $f_1$ is convex, Jensen's inequality~\cite[Theorem 4.2.1]{garling}
implies that
\begin{equation}
\label{eq:jensencoord}
  f_1(\E[Q_1 \mid Q_1 \ne 0]) \le \E[f_1(Q_1) \mid Q_1 \ne 0].
\end{equation}
Comparing the expressions for 
$\E[D(Q') \mid Q_1 \neq 0]$ and $\E[D(Q) \mid Q_1 \neq 0]$,
we see that \ineqref{eq:jensencoord} is equivalent to
\begin{equation}
\label{eq:jensencoord2}
\E[D(Q') \mid Q_1 \neq 0] \le \E[D(Q) \mid Q_1 \neq 0].
\end{equation}
Now
\begin{align*}
\E[D(Q')] =& \Pr(Q'_1 = 0) \E[D(Q') \mid Q'_1 = 0] \\
         & + \Pr(Q'_1 \ne 0) \E[D(Q') \mid Q'_1 \ne 0] \\
          =& \Pr(Q_1 = 0) \E[D(Q) \mid Q_1 = 0] \\
         & + \Pr(Q_1 \ne 0) \E[D(Q') \mid Q_1 \ne 0],
\end{align*}
while
\begin{align*}
\E[D(Q)] = & \Pr(Q_1 = 0) \E[D(Q) \mid Q_1 = 0] \\
          & + \Pr(Q_1 \ne 0) \E[D(Q) \mid Q_1 \ne 0].
\end{align*}
Therefore, \ineqref{eq:jensencoord2} implies that
\begin{equation}
\label{eq:coord}
 \E[D(Q')] \le \E[D(Q)].
\end{equation}
But $Q$ is efficient, so \ineqref{eq:coord}
must in fact be an equality.
This forces \ineqref{eq:jensencoord2} to be an equality,
which in turn forces
\ineqref{eq:jensencoord} to be an equality.
But since $f_1$ is \emph{strictly} convex,
Jensen's inequality~\cite[Theorem 4.2.1]{garling} is an equality
only if the random variable $(Q_1 \mid Q_1 \ne 0)$ is concentrated
on its mean value.

\end{proof}

Coordinate concentration allows us to further reduce the task of
finding efficient unbiased $m$-sparsifications to an $n$-dimensional problem, as follows.
Let $q_i$ be as in the statement of \cref{lem:coord-concentration},
and define $s_i \defeq \Pr(Q_i \ne 0)$.
Then the quantity we seek to minimize is
\begin{align*}
\E[D(Q)] =&
\sum_{i=1}^n \bigl(\Pr(Q_i = 0)\cdot f_i(0) \\
           & \quad \quad +  \Pr(Q_i \ne 0) \cdot \E[f_i(Q_i) \mid Q_i \ne 0]
    \bigr) \\
   =& \sum_{i=1}^n \bigl( (1-s_i) f_i(0) + s_i f_i(q_i) \bigr).
\end{align*}
It follows directly from the definitions of $s_i$ and~$q_i$
that $s_i q_i = \E[Q_i]$; on the other hand, unbiasedness means
that $\E[Q_i] = p_i$.  In the proof of
\cref{lem:coord-concentration}, we noted that $s_i > 0$,
so we may replace $q_i$ with $p_i/s_i$.  That is, we seek to
minimize
\begin{equation*}
   \sum_{i=1}^n \bigl( (1-s_i) f_i(0) + s_i f_i(p_i/s_i) \bigr).
\end{equation*}
If we replace each function $f_i(x)$ with a constant shift
$f_i(x) - c_i$ (where $c_i$ can depend on~$p$ but not on~$x$),
then the value of the divergence just changes by a constant,
which does not affect the optimization problem we are trying
to solve.  So by setting $c_i \defeq f_i(0)$, we may assume
without loss of generality that $f_i(0) = 0$ for all~$i$.
Hence we are reduced to finding unbiased $m$-sparsifications~$Q$
that minimize
\begin{equation*}
\sum_{i=1}^n s_i f_i(p_i/s_i).
\end{equation*}
Now, the sum of the $s_i$ is the expected number of nonzero
entries of~$Q$, which by definition is at most~$m$.
Therefore, we are led to consider the following optimization
problem.

\framed{\clabel{opt:ip}{IPO}%
   \textbf{Problem.} \emph{Inclusion Probability Optimization (\ref*{opt:ip}).}\\
   For strictly convex $f_i$ with $f_i(0) = 0$,
  \begin{gather*}
  \text{minimize}\quad \sum_{i=1}^n s_i f_i\left(\frac{p_i}{s_i}\right) \quad \text{subject to}\\
  \begin{alignedat}{2}
    &\sum_{i=1}^n s_i \le m\\
    \text{and} \qquad  &0 < s_i \leq 1 &\quad&\text{ for all $i$.}
  \end{alignedat}
  \end{gather*}
}

 In Appendix \ref{sec:optimization}, we solve \ref{opt:ip}.
We show in particular that the optimal solution satisfies
$\sum_i s_i = m$.  We claim that we thereby characterize all
efficient unbiased $m$-sparsifications.
Why?  Well, we have just argued
that given any efficient unbiased
$m$-sparsification~$Q$, the quantities
$s_i \defeq \Pr(Q_i \ne 0)$ must yield an optimal solution to \ref{opt:ip}.
Conversely, given any optimal solution $\{s_i\}$ to \ref{opt:ip},
we saw in \autoref{sec:marginals} that
the conditions $\sum_i s_i = m$ and $0 < s_i \le 1$
imply that it is possible to sample with the specified
marginals $\{s_i\}$.
Each way of sampling with the specified marginals
completely specifies
a unique unbiased $m$-sparsification of~$p$,
which is guaranteed to be efficient since the $\{s_i\}$
constitute an optimal solution to \ref{opt:ip}.

\subsection{ Solving \ref*{opt:ip}}

\label{sec:optimization}
 As usual, we regard $p\in \R_{>0}^n$ as fixed. Let
\begin{equation*}
F(s) \defeq \sum_{i=1}^n s_i f_i\biggl(\frac{p_i}{s_i}\biggr)
\end{equation*}
be the function we are seeking to minimize.
We show that $F$ is strictly convex\footnote{In fact,
$F$ is an \emph{$f$-divergence}~\cite{amari},
and it is a standard fact that the (strict) convexity of~$f$ implies
the (strict) convexity of~$F$, but we give a proof anyway since it is short.}.
By direct computation,
\begin{equation}
\frac{\partial F(s)}{\partial s_i} = f_i\left(\frac{p_i}{s_i}\right)
   - \left(\frac{p_i}{s_i}\right)f'_i\left(\frac{p_i}{s_i}\right). \label{eqn:g}
\end{equation}
Then
\begin{align*}
\frac{\partial^2 F(s)}{\partial s_i^2} &=
   -\frac{p_i}{s_i^2}f'_i\left(\frac{p_i}{s_i}\right)
   +\frac{p_i}{s_i^2}f'_i\left(\frac{p_i}{s_i}\right)
   + \left(\frac{p_i^2}{s_i^2}\right)f''_i\left(\frac{p_i}{s_i}\right) \\
   &= \left(\frac{p_i^2}{s_i^2}\right)f''_i\left(\frac{p_i}{s_i}\right) > 0,
\end{align*}
where the final inequality follows because $f_i$ is strictly convex
and $s_i > 0$.
So the Hessian is a diagonal matrix with strictly positive
entries on the diagonal, and $F$ is (strictly) convex.

 Motivated by \autoref{eqn:g}, define $g_i(x) \defeq x f_i'(x) - f_i(x)$, so that $\pd[F]{s_i}(s) = -g_i(p_i/s_i)$.
A similar calculation to the one above shows that $g_i$ is a strictly increasing function of $x > 0$. 
Moreover, $g_i(0) = 0$ because $f_i(0) = 0$,
so $g_i(x) > 0$ for all $x>0$.

The constraint that $s_i > 0$ is slightly awkward to deal with.
Our approach is to pick some small $\epsilon > 0$ and replace
the constraint $s_i > 0$ with $s_i \ge \epsilon$.
We then show that for all sufficiently small~$\epsilon$,
all optimal solutions are independent of~$\epsilon$
and do not lie on $s_i = \epsilon$.
Any feasible solution with $s_i > 0$ will be feasible
for some $\epsilon > 0$, and hence its objective value
cannot exceed the optimal value.

Following the standard recipe for convex
optimization~\cite[Chapter 5]{boyd-vandenberghe},
we define the Lagrangian
\begin{align*}
\mathcal{L}(s, \mu,\nu, \lambda) = & F(s) + \sum_{i=1}^n \mu_i(s_i - 1)
  + \sum_{i=1}^n \nu_i (\epsilon - s_i) \\
  & + \lambda \biggl(-m + \sum_{i=1}^n s_i\biggr),
\end{align*}
where $\mu_i$, $\nu_i$, and~$\lambda$ are Lagrange multipliers.
Let $\vec{e}_i$ denote the $i$-th unit vector,
and let $\ones \defeq \sum_i \vec{e}_i$.  Then
\begin{equation*}
\nabla \mathcal{L} = 
(\nabla F)(s) + \sum_{i=1}^n \mu_i \vec{e}_i
  - \sum_{i=1}^n \nu_i \vec{e}_i + \lambda \ones.
\end{equation*}
At an optimal point, the KKT conditions are
(in addition to the condition that an optimal point be feasible)
\begin{align*}
 -g_i\biggl(\frac{p_i}{s_i}\biggr) + \mu_i - \nu_i + \lambda &= 0 \\
  \mu_i, \nu_i, \lambda &\ge 0\\
    \mu_i(s_i - 1)  & = 0\\
 \nu_i(\epsilon - s_i)  & = 0\\
  \lambda \biggl(-m + \sum_{i=1}^n s_i\biggr) &= 0
\end{align*}

Given a proposed solution~$s$, let $H \defeq \{i \colon s_i = 1\}$
(the \emph{heavy} indices), let $E \defeq \{i \colon s_i = \epsilon\}$
(the \emph{epsilon} indices), and let $L$ denote the remaining
(\emph{light}) indices.  We first claim that $L \ne \varnothing$.
Suppose to the contrary that $L = \varnothing$.
Now $m < n$ and $\sum_i s_i \le m$, so $H$ cannot comprise
\emph{all} the indices, and the remaining indices must be in~$E$.
But for all sufficiently small positive~$\epsilon$,
$\sum_i s_i$ cannot be exactly equal to an integer~$m$,
so $\lambda$ is forced to be zero.  For any $i\in E$,
we must have $\mu_i = 0$, so $g_i(p_i/\epsilon) + \nu_i = 0$,
which implies that
\begin{equation}
\label{eq:gpi}
 g\biggl(\frac{p_i}{\epsilon}\biggr)  = -\nu_i  \le 0.
\end{equation}
But as we observed earlier, $g_i(p_i/\epsilon) > 0$ since
$p_i > 0$. This contradiction shows that $L \ne \varnothing$.

For $i\in L$, $\mu_i = \nu_i = 0$, so
\begin{equation*}
-g_i\biggl(\frac{p_i}{s_i}\biggr) + \lambda = 0.
\end{equation*}
Two inferences are immediate.
First, $\lambda = g_i(p_i/s_i) > 0$, and hence
\begin{equation*}
\sum_{i=1}^n s_i = m.
\end{equation*}
It follows that there cannot be any $\epsilon$
contribution to $\sum_i s_i$, so $E = \varnothing$.
Second, $g_i(p_i/s_i) = \lambda$ is constant across all $i\in L$.
By definition of~$L$, we must have $1 > s_i = p_i/g_i^{-1}(\lambda)$ for all $i\in L$.

 Conversely, for $i\in H$ we must have
\begin{equation*}
  g_i(p_i/s_i) - \lambda = \mu_i \geq 0
\end{equation*}
so $1 = s_i \leq p_i / g_i^{-1}(\lambda)$ (since $g_i$ is increasing). That is, for all $i$, we have $s_i = \min(1, p_i/g_i^{-1}(\lambda))$.

 The final constraint which we must satisfy (since $\lambda > 0$) is
\begin{equation}
  m = \sum_i s_i = \sum_i \min(1, p_i/g_i^{-1}(\lambda)). \label{eqn:lambda}
\end{equation}
The right hand side is a continuous, decreasing function of $\lambda$ with range $(0, n]$. Furthermore, it is strictly decreasing except where it is equal to $n > m$. Therefore there is a unique $\lambda > 0$ solving \autoref{eqn:lambda} (which can easily be found by, say, binary search), which yields our desired optimum.

\end{document}